\newcommand{\shortversion}[1]{}
\newcommand{\longversion}[1]{#1}

\shortversion{
  \documentclass{llncs}
  \pagestyle{plain}

}
\longversion{
  \documentclass[10pt,usletter]{article}
  \usepackage[totalwidth=450pt,totalheight=640pt]{geometry}
  \date{}
  \usepackage{amsthm}
}

\shortversion{
\newenvironment{myquote}{\list{}{\leftmargin=\parindent\rightmargin=0in\topsep=3pt}\item[]}{\endlist}
}
\longversion{
\newenvironment{myquote}{\begin{center}
    \begin{minipage}{.80\linewidth}}{\end{minipage}\end{center}}
}

\usepackage{times}
\usepackage{booktabs}
\usepackage{amsfonts}
\usepackage{amsmath}
\usepackage{amssymb}
\usepackage{url}
\usepackage{mathtools}
\usepackage{datetime}
\usepackage[small]{caption}
\usepackage{subcaption}
\usepackage{pgfplots}
\usepackage{hyperref}
\usepackage[lined,boxed,commentsnumbered,noend]{algorithm2e}

\usepackage{tikz}
\usetikzlibrary{arrows,backgrounds,fit}

\newcommand{\SB}{\{\,}%
\newcommand{\SM}{\;{:}\;}%
\newcommand{\SE}{\,\}}%
\newcommand{\SBs}{\{}%
\newcommand{\SEs}{\}}%

\newcommand{\NP}{\text{\normalfont NP}}
\newcommand{\coNP}{\text{\normalfont co-NP}}
\newcommand{\FPT}{\text{\normalfont FPT}}

\newcommand{\W}[1]{\text{\normalfont W[#1]}}
\newcommand{\paraNP}{\text{\normalfont para-NP}}

\newcommand{\CCC}{\mathcal{C}}

\newcommand{\Card}[1]{|#1|}
\newcommand{\rCard}[1]{||#1||}

\newcommand{\mtext}[1]{\text{\normalfont #1}}

\newcommand{\ol}[1]{\overline{#1}}
\newcommand{\kleene}{\ensuremath{^*}}

\newcommand{\CNF}{\mtext{\sc CNF}}
\newcommand{\SAT}{\mtext{\sc SAT}}

\newcommand{\Horn}{\mtext{\sc Horn}}
\newcommand{\DefHorn}{\mtext{\sc Def}\-\mtext{\sc Horn}}
\newcommand{\NUHorn}{\mtext{\sc NU}\-\mtext{\sc Horn}}
\newcommand{\Krom}{\mtext{\sc Krom}}
\newcommand{\VO}{\mtext{\sc VO}}

\newcommand{\Var}[1]{\mtext{Var(\ensuremath{#1})}}
\newcommand{\Lit}[1]{\mtext{Lit(\ensuremath{#1})}}
\newcommand{\impl}[1]{\mtext{impl}(\ensuremath{#1})}

\newcommand{\kbackbone}{\mtext{{\sc Local}-}\allowbreak{}\mtext{{\sc Back}}\-\mtext{{\sc bone}}}
\newcommand{\iterativebackbone}{\mtext{{\sc Iterative}-}\allowbreak\mtext{{\sc Local}-}\allowbreak{}\mtext{{\sc Backbone}}}
\newcommand{\MCC}{\mtext{{\sc Multi}}\-\mtext{{\sc colored}-}\allowbreak{}\mtext{{\sc Clique}}}

\newcommand{\smallunsatsubset}{\mtext{{\sc Small}-}\allowbreak{}\mtext{{\sc Unsatisfi}}\-\mtext{{\sc able}-}\allowbreak{}\mtext{{\sc Subset}}}

\newcommand{\shorthyperpath}{\mtext{{\sc Short}-}\allowbreak{}\mtext{{\sc Hyperpath}}}

\newcommand{\fptreduction}{\leq_{\mtext{fpt}}}

\allowdisplaybreaks

\def\hy{\hbox{-}\nobreak\hskip0pt}

\shortversion{
  
  \renewenvironment{proof}{\vspace{-2mm}\begin{pf}}{\qed\end{pf}}

  \newtheorem{observation}[proposition]{\textbf{Observation}}
}

\longversion{
  \newtheorem{theorem}{Theorem}
  \newtheorem{lemma}{Lemma}
  \newtheorem{proposition}{Proposition}
  \newtheorem{corollary}{Corollary}
  \newtheorem{observation}{Observation}  
  \theoremstyle{remark}
  \newtheorem{example}{Example}
}

\renewenvironment{example}{\begin{ex}}{\hfill
    $\dashv$\end{ex}\medskip}

\tikzstyle{vertex}=[draw=none]
\tikzstyle{edge} = [draw,->]

\DeclareRobustCommand{\DE}[3]{#2}

\begin{document}

\shortversion{
\title{Local Backbones}
\author{
Ronald de Haan\inst{1}\thanks{Supported by the European Research Council (ERC), project COMPLEX REASON, 239962.} \and
Iyad Kanj\inst{2} \and
Stefan Szeider\inst{1}$^{\star}$
}
\institute{
Institute of Information Systems, Vienna University of Technology,
Vienna, Austria \and
School of Computing, DePaul University, Chicago, IL
}
}

\longversion{
\title{Local Backbones}
\author{Ronald de Haan$^{1}$\thanks{Supported by the European Research Council (ERC), project COMPLEX REASON, 239962.} \and
Iyad Kanj$^2$ \and
Stefan Szeider$^{1\hspace{1pt}*}$}
\date{ \small ${}^1$ Institute of Information Systems, Vienna University of Technology,
Vienna, Austria\\
\small ${}^2$ School of Computing, DePaul University, Chicago, IL
}
}

  \maketitle

\longversion{\thispagestyle{empty}}
\shortversion{\pagestyle{empty}}
 
  \begin{abstract}
    A backbone of a propositional CNF formula is a variable whose
    truth value is the same in every truth assignment that satisfies the formula.
    The notion of backbones for CNF formulas has been studied in various contexts. 
    In this paper, we introduce local variants of backbones, and study the
    computational complexity of detecting them. In
    particular, we consider $k$\hy backbones, which are backbones for
    sub-formulas consisting of at most $k$ clauses, and iterative $k$\hy backbones,
    which are backbones
    that result after repeated instantiations of $k$\hy backbones.
    We determine the parameterized complexity of
    deciding whether a variable is a $k$\hy backbone or an
    iterative~$k$\hy backbone for various restricted
    formula classes, including Horn, definite Horn, and Krom.  We also
    present some first empirical results regarding backbones for
    CNF-Satisfiability (SAT). The empirical results we obtain show that a large
    fraction of the backbones of structured SAT instances are local, in
    contrast to random instances, which appear to have few local backbones.
  \end{abstract}

%%% Introduction
\section{Introduction}

A \emph{backbone} of a propositional
formula $\varphi$ is a variable whose truth value is the same for all
satisfying assignments of $\varphi$.
The term originates in computational
physics~\cite{SchneiderEtal96}, and the notion of backbones has been
studied for \SAT{} in various contexts.
Backbones have also been
considered in other contexts (e.g., knowledge compilation~\cite{DarwicheMarquis02})
and for other combinatorial problems~\cite{SlaneyWalsh01}.  If
a backbone and its truth value are known, then we can simplify the
formula without changing its satisfiability, or the number of
satisfying assignments. Therefore, it is desirable to have an
efficient algorithm for detecting backbones. In general, however, the
problem of identifying backbones is coNP-complete
(this follows from the fact that a literal $l$ is enforced by a formula
$\varphi$ if and only if $\varphi \wedge \neg l$ is unsatisfiable).

A variable can be a backbone because of \emph{local properties} of the
formula (such backbones we call \emph{local backbones}).
As an extreme example consider a CNF formula that contains a
unit clause. In this case we know that the variable appearing in the unit
clause is a backbone of the formula. More generally, we define the
\emph{order} of a backbone $x$ of a CNF formula $\varphi$ to be the
cardinality of a smallest subset $\varphi'\subseteq \varphi$ such that $x$
is a backbone of $\varphi'$, and we refer to backbones of order $\leq k$
as \emph{$k$\hy backbones}.  Thus, unit clauses give rise to
1-backbones.

A natural generalization of $k$\hy backbones are variables whose truth
values are enforced by repeatedly assigning $k$\hy backbones to their
appropriate truth value and simplifying the formula according to
this assignment. %%% TODO: reformulate this sentence?
We call variables that are assigned by this
iterative process \emph{iterative $k$\hy backbones}
(for a formal definition, see Section~\ref{sec:logic}).
For instance, iterative $1$\hy backbones are exactly those variables whose
truth values are enforced by unit propagation. The \emph{iterative
  order} of a backbone $x$ is the smallest $k$ such that $x$ is an
iterative $k$\hy backbone.

\paragraph{Finding Local Backbones}
For every constant $k$, we can clearly identify all $k$\hy backbones
and iterative $k$\hy backbones of a CNF formula $\varphi$ in polynomial
time by simply going over all subsets of $\varphi$ of size at most $k$
(and iterating this process if necessary). However, if~$\varphi$ consists
of $m$ clauses, then this brute-force search requires us to consider
at least $m^k$ subsets, which is impractical already for small values
of~$k$. It would be desirable to have an algorithm that detects
(iterative) $k$\hy backbones in time $f(k)\rCard{\varphi}^c$ where $f$ is a function,
$\rCard{\varphi}$ denotes the length of the formula, and $c$ is a constant. An
algorithm with such a running time would render the problem
\emph{fixed-parameter tractable} with respect to
parameter~$k$~\cite{DowneyFellows99}. In this paper we study the
question of whether the identification of (iterative) $k$\hy backbones
of a CNF formula is fixed-parameter tractable or not, considering
various restrictions on the CNF formula. We therefore define the
following template for parameterized problems, where $\CCC$ is an
arbitrary class of CNF formulas.

\begin{myquote}
  $\kbackbone[\CCC]$
  
  \emph{Instance:} a CNF formula $\varphi\in \CCC$, a variable $x$ of
  $\varphi$, and an integer $k\geq 1$.

  \emph{Parameter:} The integer $k$.

  \emph{Question:} Is $x$ a $k$\hy backbone of $\varphi$?
\end{myquote}
The problem $\iterativebackbone$ is defined similarly. 
It is not hard to see that $\kbackbone[\CCC]$ is closely related to
the problem of finding a small unsatisfiable subset of a CNF
formula (this is proven below in Lemmas~\ref{lem:1}~and~\ref{lem:2}).
More precisely, for every class $\CCC$, the problem
$\kbackbone[\CCC]$ has the same parameterized complexity as the
following problem, studied by Fellows et
al.~\cite{FellowsSzeiderWrightson06}. 
\begin{myquote}
  $\smallunsatsubset[\CCC]$
  
  \emph{Instance:} a CNF formula $\varphi\in \CCC$, and an integer $k\geq 1$.

  \emph{Parameter:} The integer $k$.

  \emph{Question:} Is there an unsatisfiable subset $\varphi'\subseteq
  \varphi$ consisting of at most $k$ clauses?
\end{myquote}
This problem is of relevance also for classes $\CCC$ for which the
satisfiability is decidable in polynomial time. For instance, given an
inconsistent knowledge base in terms of an unsatisfiable set of Horn
clauses, one might want to detect the cause for the inconsistency in
terms of a small unsatisfiable subset.

\paragraph{Results}

We draw a detailed parameterized complexity map of the
considered problems $\kbackbone[\CCC]$, $\iterativebackbone[\CCC]$,
and $\smallunsatsubset[\CCC]$, for various classes $\CCC$.
Table~\ref{table:results} provides an overview of our complexity
results ($\FPT$ indicates that the problem is fixed-parameter
tractable, $\W{1}$-hardness indicates strong evidence that the problem
is not fixed-parameter tractable; see Section~\ref{sec:prelim-pc} for
details).
\begin{table}[tb]
\newcommand{\wh}{$\W{1}$-h}
\newcommand{\wc}{$\W{1}$-c}
  \centering

  \begin{tabular}{@{}l@{\qquad}l@{\qquad}l@{}} \toprule
    $\CCC $ & $\kbackbone[\CCC]$ &  $\iterativebackbone[\CCC]$ \\ \midrule
    $\CNF$                 & \wc \hfill(Thm~\ref{thm:bb-defhorn}) &  \wh \hfill(Cor~\ref{cor:ibb-horn}) \\
    $\DefHorn$  & \wc \hfill(Thm~\ref{thm:bb-defhorn}) &  P \hfill(Thm~\ref{thm:ibb-defhorn}) \\
    $\NUHorn$  & \wc \hfill(Thm~\ref{thm:bb-horn-nuhorn}) &  \wh \hfill(Cor~\ref{cor:ibb-horn}) \\
    $\Krom$                & P \hfill(Prop~\ref{prop:bb-krom}) &  P \hfill(Thm~\ref{thm:ibb-krom}) \\ 
    $\VO_d$                & FPT \hfill(Thm~\ref{thm:bb-vo}) &  FPT \hfill(Thm~\ref{thm:ibb-vo}) \\  \bottomrule
  \end{tabular}

  \vspace{5mm}
  \caption{Map of parameterized complexity results.
  (The classes $\CCC$ of formulas are defined in Section~\ref{sec:logic}.)}
  \shortversion{\vspace{-30pt}}
  \label{table:results}
\end{table}
\longversion{

}
It is interesting to observe that the non-iterative problems tend to
be at least as hard as the iterative problems.
%Somewhat
%surprising is the $\W{1}$\hy hardness of $\kbackbone[\Krom]$ and
%$\smallunsatsubset[\Krom]$ (which also implies the $\NP$\hy hardness of
%the unparameterized versions of these problems).
%On the one hand, this seems to contrast with the
%fact that a shortest tree-like resolution refutation of an unsatisfiable Krom
%formula can be found in polynomial time
%\cite{BureshOppenheimMitchell06}.
%On the other hand, this is in line with the result
%that deciding whether a \CNF{} formula can be refuted within $k$
%resolution steps (parameterized by $k$)
%is \W{1}-complete \cite{FellowsSzeiderWrightson06}.
%
The polynomial time solvability of finding iterative local backbones
in definite Horn formulas is also interesting,
especially in the light of the intractability of the corresponding problem
of finding (non-iterative) local backbones.

We also provide some first empirical results on the distribution of
local backbones in some benchmark SAT instances. We consider
structured instances and random instances. For the structured instances
that we consider we observe that a large
fraction of the backbones are of relatively small iterative order. In
contrast, the backbones of the random instances that we consider are of
large iterative order. The results suggest that the distribution of
the iterative order of backbones might be an indicator for a hidden
structure in SAT instances.

\paragraph{Related Work}
The notion of backbones has initially been studied in
the context of optimization problems in computational physics~\cite{SchneiderEtal96}.
The notion has later been applied to several
combinatorial problems~\cite{SlaneyWalsh01}, including \SAT{}.
The relation between backbones
and the difficulty of finding a solution for \SAT{}
has been studied by Kilby et al.~\cite{KilbySlaneyThiebauxWalsh05},
by Parkes~\cite{Parkes97} and by Slaney and Walsh~\cite{SlaneyWalsh01}.
The complexity of finding backbones
has been studied theoretically by Kilby et al.~\cite{KilbySlaneyThiebauxWalsh05}.
The notion of backbones has also been used
for improving SAT solving algorithms by
Dubois and Dequen~\cite{DuboisDequen01}
and by Hertli et al.~\cite{HertliMoserScheder11}.
The problem of identifying unsatisfiable subsets
of size at most $k$ has been considered by
Fellows et al.~\cite{FellowsSzeiderWrightson06},
who proved that this problem \longversion{(parameterized on $k$)}
is \W{1}-complete.
Furthermore, they showed by the same reduction
that finding a $k$-step resolution refutation for a given formula
is \W{1}-complete as well.
Related notions of locally enforced literals have also been studied,
including a notion of generalized unit-refutation
\cite{GwynneKullmann13,Kullmann99}.

\shortversion{
\paragraph{Full Version}
Because of space constraints some proofs have been
omitted or shortened. Detailed proofs can be found in the full
version, available at  \url{arxiv.org/abs/1304.5479}.
}

%%% Preliminaries
\section{Preliminaries}

\subsection{CNF Formulas, Unsatisfiable Subsets and Local Backbones}
\label{sec:logic}

A \emph{literal} is a propositional variable $x$ or a negated variable $\neg x$.
The \emph{complement} $\overline{x}$ of a positive literal $x$ is $\neg x$,
and the complement $\overline{\neg x}$ of a negative literal $\neg x$ is $x$.
A \emph{clause} is a finite set of literals, not containing a complementary pair $x$, $\neg x$.
A \emph{unit clause} is a clause of size 1.
We let $\bot$ denote the empty clause.
A \emph{formula} in conjunctive normal form (or \CNF{} formula)
is a finite set of clauses.
We define the \emph{length} $\rCard{\varphi}$ of a formula~$\varphi$ to be $\sum_{c \in \varphi} \Card{c}$;
the number of clauses of $\varphi$ is denoted by $\Card{\varphi}$.
A formula $\varphi$ is a $k$\hy \CNF{} formula if the size of each of its clauses is at most $k$.
A 2-\CNF{} formula is also called a Krom formula.
A clause is a \emph{Horn clause} if it contains at most one positive literal.
A Horn clause containing exactly one positive literal is a \emph{definite Horn clause}.
Formulas containing only Horn clauses are called \emph{Horn formulas}.
\emph{Definite Horn formulas} are defined analogously.
We denote the class of all Krom formulas by \Krom{},
the class of all Horn formulas by \Horn{}
and the class of all definite Horn formulas by \DefHorn{}.
We let \NUHorn{} denote the class of Horn formulas
not containing unit clauses
(such formulas are always satisfiable).
Let $d$ be an integer. The class of \CNF{} formulas such
that each variable occurs at most $d$ times
is denoted by $\VO_d$.

For a \CNF{}-formula $\varphi$, the set $\Var{\varphi}$ denotes
the set of all variables $x$ such that some clause of $\varphi$ contains $x$ or $\neg x$;
the set $\Lit{\varphi}$ denotes the set of all literals $l$
such that some clause of $\varphi$ contains $l$ or $\overline{l}$.
A formula $\varphi$ is \emph{satisfiable} if there exists an assignment $\tau : \Var{\varphi} \rightarrow \SBs 0,1 \SEs$
such that every clause $c \in \varphi$ contains some variable~$x$ with $\tau(x)=1$
or some negated variable $\neg x$ with $\tau(x) = 0$
(we say that such an assigment $\tau$ satisfies~$\varphi$);
otherwise, $\varphi$ is \emph{unsatisfiable}.
$\varphi$ is \emph{minimally unsatisfiable} if~$\varphi$ is unsatisfiable
and every proper subset of $\varphi$ is satisfiable.
It is well-known that any minimal unsatisfiable CNF formula
has more clauses than variables
(this is known as Tarsi's Lemma \cite{AharoniLinial86,Kullmann99d}).
For two formulas $\varphi,\psi$,
whenever all assignments satisfying $\varphi$
also satisfy $\psi$,
we write $\varphi \models \psi$.
The reduct $\varphi|_{L}$ of a formula $\varphi$
with respect to a set of literals $L \subseteq \Lit{\varphi}$
is the set of clauses of $\varphi$ that do not contain any $l \in L$
with all occurrences of $\overline{l}$ for all $l \in L$ removed.
For singletons $L = \SBs l \SEs$, we also write $\varphi|_{l}$.
We say that a class $\CCC$ of formulas is \emph{closed under variable instantiation}
if for every $\varphi \in \CCC$ and every $l \in \Lit{\varphi}$
we have that $\varphi|_{l} \in \CCC$.
For an integer $k$, a variable $x$
is a \emph{$k$\hy backbone} of~$\varphi$,
if there exists a $\varphi' \subseteq \varphi$ such that $\Card{\varphi'} \leq k$
and either $\varphi' \models x$ or $\varphi' \models \neg x$.
A variable $x$ is a \emph{backbone} of a formula $\varphi$ if it is a
$\Card{\varphi}$-backbone.
Note that the definition of the backbone of a formula $\varphi$
that is used in some of the literature
includes all literals $l \in \Lit{\varphi}$
such that $\varphi \models l$.
For an integer $k$, a variable $x$
is an \emph{iterative $k$\hy backbone} of $\varphi$
if either (i) $x$ is a $k$\hy backbone of $\varphi$,
or (ii) there exists $y \in \Var{\varphi}$ such that
$y$ is a $k$\hy backbone of $\varphi$,
and for some $l \in \SBs y,\neg y \SEs$,
$\varphi \models l$
and $x$ is an iterative $k$\hy backbone of~$\varphi|_{l}$.

For a Krom formula $\varphi$, we let $\impl{\varphi}$ be the \emph{implication graph} $(V,E)$ of $\varphi$,
where $V = \SB x,\neg x \SM x \in \Var{\varphi} \SE$
and $E = \SB (\ol{a},b), (\ol{b},a) \SM \SBs a,b \SEs \in \varphi \SE$.
We say that a path $p$ in this graph \emph{uses a clause} $\SBs a,b \SEs$ of $\varphi$
if either one of the edges $(\ol{a},b)$ and $(\ol{b},a)$ occurs in $p$;
we say that $p$ \emph{doubly uses} this clause if both edges occur in $p$.

\subsection{Parameterized Complexity}
\label{sec:prelim-pc}

Here we introduce the relevant concepts of parameterized complexity theory.
For more details, we refer to text books on the topic~\cite{DowneyFellows99,FlumGrohe06,Niedermeier06}.
An instance of a parameterized problem is a pair $(I,k)$
where $I$ is the main part of the instance,
and $k$ is the parameter.
A parameterized problem is \emph{fixed-parameter tractable}
if instances $(I,k)$ can be solved by a deterministic algorithm
that runs in time $f(k)\Card{I}^c$,
where $f$ is a computable function of $k$,
and $c$ is a constant
(algorithms running within such time bounds are called \emph{fpt-algorithms}).
If $c = 1$, we say the problem is \emph{fixed-parameter linear}.
\FPT{} denotes the class of all fixed-parameter tractable
problems.
Using fixed-parameter tractability, many problems that are classified as
intractable in the classical setting can be shown to be tractable
for small values of the parameter.

Parameterized complexity also offers a \emph{completeness theory},
similar to the theory of \NP{}-completeness.
This allows the accumulation of strong theoretical evidence that
a parameterized problem is not fixed-parameter tractable.
Hardness for parameterized complexity classes is based on fpt-reductions,
which are many-one reductions where the parameter of one problem
maps into the parameter for the other.
More specifically, a parameterized problem $L$ is fpt-reducible to another
parameterized problem $L'$ (denoted $L \fptreduction L'$) if there is a mapping $R$
from instances of $L$ to instances of $L'$ such that
(i) $(I,k) \in L$ if and only if $(I',k') = R(I,k) \in L'$,
(ii) $k' \leq g(k)$ for a computable function $g$, and
(iii) $R$ can be computed in time $O(f(k)\Card{I}^c)$
for a computable function~$f$ and a constant $c$.

Central to the completeness theory is the hierarchy
$\FPT{} \subseteq \W{1} \subseteq \W{2} \subseteq \dotsm \subseteq \paraNP{}$.
Each intractability class \W{t} contains all parameterized problems that can be reduced
to a certain parameterized satisfiability problem
under fpt-reductions.
The intractability class \paraNP{} includes all parameterized problems
that can be solved by a nondeterministic fpt-algorithm.
Fixed-parameter tractability of any problem hard for any of these intractability classes
would imply that the Exponential Time Hypothesis fails~\cite{FlumGrohe06,ImpagliazzoPaturiZane01}
(i.e., the existence of a $2^{o(n)}$ algorithm for $n$-variable $3\SAT{}$).

%%% Local Backbones
\section{Local Backbones and Small Unsatisfiable Subsets}

The straightforward reductions in the proofs of the following two lemmas,
illustrate the close connection
between \kbackbone{} and \smallunsatsubset{}.

\begin{lemma}
\label{lem:1}
$\smallunsatsubset{} \fptreduction \kbackbone{}$.
\end{lemma}
\begin{proof}
Let $(\varphi,k)$ be an instance of \smallunsatsubset{}.
We construct an instance $(\varphi',z,k)$ of \kbackbone{},
by letting $\varphi' = \SB c \cup \SBs z \SEs \SM c \in \varphi \SE$
for some $z \not\in \Var{\varphi}$.
We claim that $(\varphi,k) \in \smallunsatsubset{}$ if and only if
$(\varphi',z,k) \in \kbackbone{}$.
%%%
\shortversion{
A complete proof of this claim can be found in the full version of the paper.
}
%%%
\longversion{

$(\Rightarrow)$ Assume $(\varphi,k) \in \smallunsatsubset{}$.
Then there exists an unsatisfiable $\varphi'' \subseteq \varphi$ with
$\Card{\varphi''} \leq k$.
Now consider $\chi = \SB c \cup \SBs z \SEs \SM c \in \varphi'' \SE$.
Clearly, $\chi \subseteq \varphi'$ and $\Card{\chi} \leq k$.
Also, since $\varphi''$ is unsatisfiable, we get $\chi \models z$.
Thus $\chi$ witnesses that $z$
is a $k$\hy backbone of $\varphi'$.

$(\Leftarrow)$ Assume $(\varphi',k,z) \in \kbackbone{}$.
Since $\neg z$ does not occur in $\varphi'$,
this means that there exists a $\varphi'' \subseteq \varphi'$ with
$\Card{\varphi''} \leq k$ such that $\varphi'' \models z$.
Now take $\chi = \SB c \backslash \SBs z \SEs \SM c \in \varphi'' \SE$.
We get that $\chi \subseteq \varphi$ and $\Card{\chi} \leq k$.
Also, we know that $\chi$ is unsatisfiable,
since otherwise it would not hold that $\varphi'' \models z$.
Therefore, $(\varphi,k) \in \smallunsatsubset{}$.
}
%%%
\end{proof}

\longversion{
\noindent The reduction in the proof of Lemma \ref{lem:2} shows
that instances $(\varphi,z,k)$ of \kbackbone{}
lead to equivalent instances of \smallunsatsubset{}
by simply taking the disjoint union
of the reducts of $\varphi$ with respect to both $z$ and $\neg z$.
}

\begin{lemma}
\label{lem:2}
$\kbackbone{} \fptreduction \smallunsatsubset{}$.
\end{lemma}
\begin{proof}
Let $(\varphi,z,k)$ be an instance of \kbackbone{}.
We construct an instance $(\psi,k)$ of \smallunsatsubset{}.
For every variable $x \in \Var{\varphi}$ we take two copies $x_1,x_2$.
For $i \in \SBs 1,2 \SEs$ we let $\varphi_i$ be a copy of $\varphi$
using the variables $x_i$.
Now we define $\psi = \varphi_1|_{z_1} \cup \varphi_2|_{\neg z_2}$.
In other words, $\psi$ is the union of two disjoint copies of the reducts of $\varphi$
with respect to $z$ and $\neg z$.
We claim that $(\varphi,z,k) \in \kbackbone{}$ if and only if
$(\psi,k) \in \smallunsatsubset{}$.
%%%
\shortversion{
A complete proof of this claim can be found in the full version of the paper.
}
%%%
\longversion{

$(\Rightarrow)$
Assume $z$ is a $k$\hy backbone of $\varphi$.
This means there exists a $\varphi' \subseteq \varphi$
with $\Card{\varphi'} \leq k$ such that
either $\varphi' \models z$ or $\varphi' \models \neg z$.
Assume without loss of generality that $\varphi' \models z$.
Then $\varphi'|_{\neg z}$ is unsatisfiable.
One can see this as follows.
Assume the contrary, i.e., that $\varphi'|_{\neg z}$ is satisfiable.
This means there is a valuation $V$ that satisfies all clauses in $\varphi'|_{\neg z}$.
Let $V'$ be a valuation for $\varphi'$ defined by
\[ V'(x) = \begin{dcases*}
  0 & if $x = z$ \\
  V(x) & otherwise. \\
\end{dcases*} \]
We show that $V$ satisfies $\varphi'$.
For each clause $c \in \varphi'$ such that $\neg z \in c$,
we clearly have that $V'$ satisfies $c$, because $V'(z) = 0$.
For all other clauses $c \in \varphi'$ with $\neg z \not\in c$,
we know $V'$ satisfies $c$, by the following argument.
Because $c' = c \backslash \SBs z \SEs \in \varphi'|_{\neg z}$,
we know that $V$ satisfies some literal in $c'$.
Therefore, we know that $V'$ satisfies $c$.
This is a contradiction to the fact that $\varphi' \models z$.
Thus, $\varphi'|_{\neg z}$ is unsatisfiable.
Furthermore, we know that $\Card{(\varphi'|_{\neg z})} \leq \Card{\varphi'} \leq k$.
Also, since $\varphi' \subseteq \varphi$,
we know that $\varphi'|_{\neg z} \subseteq \varphi|_{\neg z}$.
Then, by the fact that $\varphi_2|_{\neg z_2}$ is a copy of
$\varphi|_{\neg z}$,
we know that $(\psi,k) \in \smallunsatsubset{}$.

$(\Leftarrow)$
Assume $(\psi,k) \in \smallunsatsubset{}$.
This means there exists an unsatisfiable $\psi' \subseteq \psi$
with $\Card{\psi'} \leq k$.
Since $\psi = \varphi_1|_{z_1} \cup \varphi_2|_{\neg z_2}$
and $\varphi_1|_{z_1}$ and $\varphi_2|_{\neg z_2}$ are disjoint,
we can assume without loss of generality that
either $\psi' \subseteq \varphi_1|_{z_1}$ or $\psi' \subseteq \varphi_2|_{\neg z_2}$.
Suppose that $\psi' \subseteq \varphi_2|_{\neg z_2}$,
and the other case is similar.
We then know that there is a subset $\varphi' \subseteq \varphi$
such that $\psi'$ is a copy of $\varphi'|_{\neg z}$.
Take such a $\varphi'$ of minimal size.
We then know that there is no clause $c \in \varphi'$
such that $\neg z \in c$.
Therefore, we know that $\Card{\varphi'} = \Card{\psi'} \leq k$.
We now show that $\varphi' \models z$.
Assume the contrary, i.e., that there exists an assignment $V$
with $V(z) = 0$ that satisfies $\varphi'$.
Then this $V$ would also satisfy $\varphi'|_{\neg z}$.
From this one can straightforwardly construct an assignment that satisfies
$\psi'$, which contradicts our assumption that $\psi'$ is unsatisfiable.
Thus $\varphi'$ witnesses that $z$ is a $k$\hy backbone of $\varphi$.
}
%%%
\end{proof}

\begin{theorem}
\label{thm:1}
\kbackbone{} is \W{1}-complete.
\end{theorem}
\begin{proof}
Since \smallunsatsubset{} is \W{1}-complete~\cite{FellowsSzeiderWrightson06},
the result follows from Lemmas \ref{lem:1} and \ref{lem:2}.
\end{proof}

\section{Local Backbones of Horn and Krom Formulas}
\paragraph{Horn formulas}
Restricting the problem of finding backbones in arbitrary formulas
to Horn formulas reduces the classical complexity from \coNP{}-completeness
to polynomial time solvability.
It is a natural question whether the parameterized complexity of finding local backbones
decreases in a similar way when the problem is restricted to Horn formulas.
We will show that this is not the case.
In order to do so, we define the parameterized problem \shorthyperpath{},
show that it is \W{1}-hard,
and then provide fpt-reductions from \shorthyperpath{}.

% Horn-Hyperpaths
For a Horn formula $\varphi$
and $s,t \in \Var{\varphi}$,
we say that a subformula $\varphi' \subseteq \varphi$
is a \emph{hyperpath} from $s$ to $t$
if (i) $t = s$ or
(ii) $c = \SBs x_1,\dotsc,x_n,t \SEs \in \varphi'$
and $\varphi' \backslash c$ is a hyperpath from
$s$ to $x_i$ for each $1 \leq i \leq n$.
If $\Card{\varphi} \leq k$ then $\varphi$
is called a \emph{$k$\hy hyperpath}.
The parameterized problem \shorthyperpath{}
takes as input a Horn formula $\varphi$,
two variables $s,t \in \Var{\varphi}$ and an integer $k$.
The problem is parameterized by $k$.
The question is whether there exists a $k$\hy hyperpath
from $s$ to $t$.
For a more detailed discussion on the relation
between (backward) hyperpaths in hypergraphs
and hyperpaths as defined above, we refer to a survey
article by Gallo et al.~\cite{Gallo}.

For the hardness proof of \shorthyperpath{},
we reduce from the \W{1}-complete problem \MCC{}~\cite{FellowsHermelinRosamondVialette09}.
The \MCC{} problem takes as input a graph $G$,
some integer $k$, and a proper $k$\hy coloring $c$ of the vertices of $G$.
The problem is parameterized by $k$.
The question is whether there is a properly colored $k$\hy clique in $G$.

\begin{lemma}
\label{lem:5}
\shorthyperpath{} is \W{1}-hard,
even for instances $(\varphi,s,t,k)$ where
$\varphi \in 3\CNF{}$.
\vspace{-3mm}
\end{lemma}
\begin{proof}
We give a reduction from \MCC{}.
Let $(G,k,c)$ be an instance of \MCC{},
where $G = (V,E)$ and $V_1,\dotsc,V_k$ are
the equivalence classes of $V$
induced by the $k$\hy coloring $c$.
We construct an instance $(\varphi,s,t,k')$ of \shorthyperpath{},
where $k' = k + \binom{k}{2} + 1$ and\\[-5pt]
\[ \begin{array}{r l}
  \Var{\varphi} = & \SBs s,t \SEs \cup V \cup \SB p_{i,j} \SM 1 \leq i < j \leq k \SE ; \\
  \varphi = & \varphi_V \cup \varphi_{p} \cup \varphi_{t} ; \\
  \varphi_V = & \SB \SBs \neg s, v \SEs \SM v \in V \SE ; \\
  \varphi_p = & \SB \SBs \neg v_i, \neg v_j, p_{i,j} \SEs \SM 1 \leq i < j \leq k, v_i \in V_i, v_j \in V_j, \SBs v_i, v_j \SEs \in E \SE ; \\
  \varphi_t = & \SBs \SB \neg p_{i,j} \SM 1 \leq i < j \leq k \SE \cup \SBs t \SEs \SEs.
\end{array} \]
\sloppypar \noindent This construction is illustrated
for an example with $k=3$ in Figure~\ref{fig:hyperpathclique}.
We claim that $(G,k,c) \in \MCC{}$ if and only if $(\varphi,s,t,k') \in \shorthyperpath{}$.
%%%
\shortversion{
A complete proof of this claim can be found in the full version of the paper.
}
%%%
\longversion{

$(\Rightarrow)$
Assume $(G,k,c) \in \MCC{}$.
Then there exists a clique $V'$ of $G$
with $\Card{V \cap V_i} = 1$ for all $1 \leq i \leq k$.
We construct a $k'$-hyperpath $\varphi'$ from $s$ to $t$.
We define:
\[ \begin{array}{r l}
  \varphi_{V'} = & \SB (\SBs s \SEs,v) \SM v \in V' \SE \cup \varphi_t\ \cup \\
  & \SB ( \SBs v_i,v_j \SEs, p_{i,j}) \SM 1 \leq i < j \leq k, v_i \in V_i \cap V', v_j \in V_j \cap V', \SBs v_i, v_j \SEs \in E \SE
\end{array} \]
It is straightforward to verify that
$\varphi_{V'}$ is a $k'$-hyperpath from $s$
to $t$.

$(\Leftarrow)$
Assume $(\varphi,s,t,k') \in \shorthyperpath{}$.
Then there exists a $k'$-hyper\-path $\varphi'$
from $s$ to $t$.
We know that $\varphi_t \subseteq \varphi'$,
since $\varphi_t$ contains the unique clause in $\varphi$ with $t$ occurring positively.
Since $\Card{\varphi'} \leq k'$, we know that
in order for $\varphi'$ to be a hyperpath from $s$ to $t$,
we have $\Card{\varphi_V \cap \varphi'} = k$ and $\Card{\varphi_p \cap \varphi'} = \binom{k}{2}$.
It is then straightforward to verify that
the set $V' = \SB v \in V \SM \SBs \neg s, v \SEs \in \varphi' \SE$
witnesses that $G$ has a $k$\hy clique
containing one node in each $V_i$.
}
%%%

To see that clauses of size at most $3$ in the hyperpath
suffice, we slightly adapt the reduction.
The only clause we need to change is the
single clause $e \in \varphi_t$.
This clause $e$ is of the form $\SBs \neg p_1, \dotsc, \neg p_m, t \SEs$,
for $m = \binom{k}{2}$.
We introduce new variables $v_1,\dotsc,v_m$
and replace $e$ by the $m + 1$ many clauses
$\SBs \neg p_1, v_1 \SEs$,
$\SBs \neg v_{i-1}, \neg p_i, v_i \SEs$ for all $1 < i \leq m$
and $\SBs \neg v_m, t \SEs$.
Clearly, the resulting Horn formula only has clauses
of size at most $3$.
This adapted reduction works with the exact same line of reasoning
as the reduction described above,
with the only change that $k' = k + 2\binom{k}{2} + 1$.

Note that even the slightly stronger claim holds
that $G$ has a properly colored $k$\hy clique if and only if
there exists a (subset) minimal $k'$-B-hyperpath $\varphi' \subseteq \varphi$
for which we have $\Card{\varphi'} = k'$.
\end{proof}

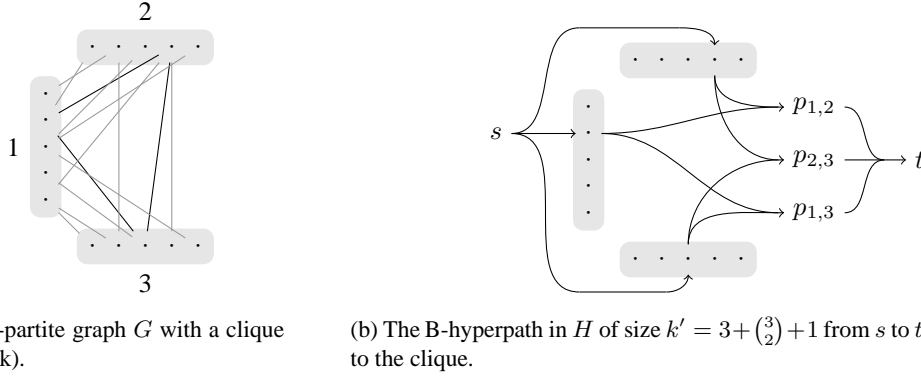
\begin{figure}[tb]
  \centering
  \begin{subfigure}[b]{0.3\textwidth}
  {
    \centering
    \begin{tikzpicture}[node distance=10pt]
      % NODES
    
      % A
      \node[] (a1) {$\cdot$};
      \node[below of=a1] (a2) {$\cdot$};
      \node[below of=a2] (a3) {$\cdot$};
      \node[below of=a3] (a4) {$\cdot$};
      \node[below of=a4] (a5) {$\cdot$};
    
      % B
      \node[above right of=a1, node distance=25pt] (b1) {$\cdot$};
      \node[right of=b1] (b2) {$\cdot$};
      \node[right of=b2] (b3) {$\cdot$};
      \node[right of=b3] (b4) {$\cdot$};
      \node[right of=b4] (b5) {$\cdot$};
    
      % C
      \node[below right of=a5, node distance=25pt] (c1) {$\cdot$};
      \node[right of=c1] (c2) {$\cdot$};
      \node[right of=c2] (c3) {$\cdot$};
      \node[right of=c3] (c4) {$\cdot$};
      \node[right of=c4] (c5) {$\cdot$};
    
      % EDGES
    
      % clique
      \path[draw,-] (a2) -- (b4);
      \path[draw,-] (b4) -- (c3);
      \path[draw,-] (a2) -- (c3);
    
      % other edges
      \path[draw,-,gray!80] (a2) -- (b1);
      \path[draw,-,gray!80] (a1) -- (b2);
      \path[draw,-,gray!80] (a3) -- (b3);
      \path[draw,-,gray!80] (a3) -- (b5);
      \path[draw,-,gray!80] (a5) -- (b4);
      \path[draw,-,gray!80] (a5) -- (c1);
      \path[draw,-,gray!80] (a5) -- (c2);
      \path[draw,-,gray!80] (a3) -- (c5);
      \path[draw,-,gray!80] (a4) -- (c3);
      \path[draw,-,gray!80] (b2) -- (c2);
      \path[draw,-,gray!80] (b4) -- (c4);
      
      % BACKGROUNDS
      \begin{pgfonlayer}{background}
        \node [fill = gray!20, fit=(a1) (a2) (a3) (a4) (a5), rounded corners, inner sep=1pt, label=left:{1}] {};
        \node [fill = gray!20, fit=(b1) (b2) (b3) (b4) (b5), rounded corners, inner sep=1pt, label=above:{2}] {};
        \node [fill = gray!20, fit=(c1) (c2) (c3) (c4) (c5), rounded corners, inner sep=1pt, label=below:{3}] {};
      \end{pgfonlayer}
      
    \end{tikzpicture}
    \caption{A 3-partite graph $G$ with a clique (in black).}
  }
  \end{subfigure}
  \hspace{20pt}
  \begin{subfigure}[b]{0.60\textwidth}
  {
    \centering
    \begin{tikzpicture}[node distance=10pt]
      %%% NODES
    
      % A
      \node[] (a1) {$\cdot$};
      \node[below of=a1] (a2) {$\cdot$};
      \node[below of=a2] (a3) {$\cdot$};
      \node[below of=a3] (a4) {$\cdot$};
      \node[below of=a4] (a5) {$\cdot$};
    
      % B
      \node[above right of=a1, node distance=25pt] (b1) {$\cdot$};
      \node[right of=b1] (b2) {$\cdot$};
      \node[right of=b2] (b3) {$\cdot$};
      \node[right of=b3] (b4) {$\cdot$};
      \node[right of=b4] (b5) {$\cdot$};
    
      % C
      \node[below right of=a5, node distance=25pt] (c1) {$\cdot$};
      \node[right of=c1] (c2) {$\cdot$};
      \node[right of=c2] (c3) {$\cdot$};
      \node[right of=c3] (c4) {$\cdot$};
      \node[right of=c4] (c5) {$\cdot$};

      % ONE
      \node[left of=a2, node distance=35pt] (s) {$s$};
      \node[above of=a1, node distance=30pt, xshift=0pt] (sb1) {};
      \node[below of=a5, node distance=30pt, xshift=0pt] (sc1) {};
      \node[right of=sb1, node distance=-7pt] (sb2) {};
      \node[right of=sc1, node distance=-7pt] (sc2) {};
      \node[right of=sb2, node distance=45pt] (sb3) {};
      \node[right of=sc2, node distance=40pt] (sc3) {};
      \node[right of=sb3, node distance=-7pt] (sb4) {};
      \node[right of=sc3, node distance=-7pt] (sc4) {};
      
      % TWO
      \node[below right of=b5, node distance=25pt, yshift=-0pt] (yab') {};
      \node[below of=yab', node distance=20pt] (ybc') {};
      \node[below of=ybc', node distance=20pt] (yac') {};
      \node[right of=yab'] (yab) {$p_{1,2}$};
      \node[right of=ybc'] (ybc) {$p_{2,3}$};
      \node[right of=yac'] (yac) {$p_{1,3}$};
      
      % THREE
      \node[right of=ybc, node distance=30pt] (t') {};
      \node[right of=t', node distance=10pt] (t) {$t$};
      
      %%% EDGES
    
      % ONE
      \path[draw,->] (s) -- (a2);
      \path[draw,-] (s) to[out=0, in=180] (sb1);
      \path[draw,-] (s) to[out=0, in=180] (sc1);
      \path[draw,-] (sb2) -- (sb3);
      \path[draw,-] (sc2) -- (sc3);
      \path[draw,->] (sb4) to[out=0, in=90] (b4);
      \path[draw,->] (sc4) to[out=0, in=270] (c3);
    
      % TWO
      \path[draw,-] (a2) to[out=0,in=180] (yab');
      \path[draw,-] (b4) to[out=270,in=180] (yab');
      \path[draw,<-] (yab) to[xshift=-30pt] (yab');
      
      \path[draw,-] (b4) to[out=270,in=180] (ybc');
      \path[draw,-] (c3) to[out=90,in=180] (ybc');
      \path[draw,<-] (ybc) to[xshift=-30pt] (ybc');
      
      \path[draw,-] (a2) to[out=0,in=180] (yac');
      \path[draw,-] (c3) to[out=90,in=180] (yac');
      \path[draw,<-] (yac) to[xshift=-30pt] (yac');
      
      % THREE
      \path[draw,-] (yab) to[out=0,in=180] (t');
      \path[draw,-] (ybc) to[out=0,in=180] (t');
      \path[draw,-] (yac) to[out=0,in=180] (t');
      \path[draw,<-] (t) to[xshift=-30pt] (t');
      
      %%% BACKGROUNDS
      \begin{pgfonlayer}{background}
        \node [fill = gray!20, fit=(a1) (a2) (a3) (a4) (a5), rounded corners, inner sep=1pt] {};
        \node [fill = gray!20, fit=(b1) (b2) (b3) (b4) (b5), rounded corners, inner sep=1pt] {};
        \node [fill = gray!20, fit=(c1) (c2) (c3) (c4) (c5), rounded corners, inner sep=1pt] {};
      \end{pgfonlayer}
      
    \end{tikzpicture}
    \caption{The B-hyperpath in $H$ of size $k' = 3 + \binom{3}{2} + 1$ from $s$ to $t$ corresponding to the clique.}
  }
  \end{subfigure}
  \caption{Illustration of the reduction in the proof of Lemma~\ref{lem:5} for the case of a 3-colored clique.}
  \shortversion{\vspace{-10pt}}
  \label{fig:hyperpathclique}
\end{figure}

\noindent We are now in a position to prove the \W{1}-hardness of $\kbackbone[\Horn]$.
In fact, we show that finding local backbones is already \W{1}-hard
for definite Horn formulas with a single unit clause.
We also show that this hardness crucially depends on allowing
unit clauses in the formula,
since for definite Horn formulas without unit clauses the problem is trivial.
In fact, the complexity jumps to \W{1}-hardness
already when allowing a single unit clause.

\longversion{
\begin{lemma}
\label{lem:nudefhorn}
Definite Horn formulas
without unit clauses
have no backbones.
\end{lemma}
\begin{proof}
Consider the two valuations $I_{\top}$ and $I_{\bot}$,
where $I_{\top}(x) = \top$ and $I_{\bot}(x) = \bot$ for all $x \in \Var{\varphi}$.
Since $\varphi \in \DefHorn$ and $\varphi$
has no unit clauses,
we know that each clause has one positive
and at least one negative literal.
Thus both $I_{\top}$ and $I_{\bot}$ satisfy $\varphi$.
Therefore, no $x \in \Var{\varphi}$
is a backbone of $\varphi$.
\end{proof}
}

\begin{theorem}
\label{thm:bb-defhorn}
$\kbackbone[\DefHorn \cap 3\CNF]$ is \W{1}-hard,
already for instances $(\varphi,x,k)$ where
$\varphi$ has at most one unit clause.
\end{theorem}
\begin{proof}
We show \W{1}-hardness by reducing from \shorthyperpath{}.
Let $(\varphi,s,t,k)$ be an instance of \shorthyperpath{}.
We can assume that $\varphi \in 3\CNF{}$.
We construct an instance $(\psi_{\varphi},t,k')$ of \kbackbone{}.
Here $k' = k+1$.
For each $\varphi' \subseteq \varphi$ we define a formula $\psi_{\varphi'}$, by letting $\Var{\psi_{\varphi'}} = \Var{\varphi'}$
\shortversion{
and $\psi_{\varphi'} = \SBs \SBs s \SEs \SEs \cup \varphi'$.
}
\longversion{and:
\[ \begin{array}{r l}
%  \Var{\psi_{\varphi'}} = & \SB x_v \SM (A,b) \in \varphi', v \in A \cup \SBs b \SEs \SE , \\
  \psi_{\varphi'} = & \SBs \SBs s \SEs \SEs \cup \varphi'. \\
%  \SB \SBs \neg x_{a}, \neg x_{b}, x_{c} \SEs \SM \SBs \neg a, \neg b, c \SEs \in \varphi' \SE\ \cup \\
%  & \SB \SBs \neg x_{a}, x_{b} \SEs \SM \SBs \neg a, b \SEs \in \varphi' \SE .
\end{array} \]
}
Clearly $\psi_{\varphi} \in \DefHorn \cap 3\CNF$
and $\psi_{\varphi}$ has only a single unit clause.
We claim that $(\psi_{\varphi},t,k') \in \kbackbone{}$ if and only if
$(\varphi,s,t,k) \in \shorthyperpath{}$.
%%%
\shortversion{
A complete proof of this claim can be found in the full version of the paper.
}
%%%
\longversion{

$(\Rightarrow)$
Assume that $t$ is a $k'$-backbone of $\psi_{\varphi}$.
Since $\psi_{\varphi} \in \DefHorn$,
we then know that there exists a $\psi' \subseteq \psi_{\varphi}$ with
$\Card{\psi'} \leq k'$ and $\psi' \models t$.
By Lemma \ref{lem:nudefhorn}, we know that $\SBs s \SEs \in \psi'$.
Now let $\varphi' \subseteq \varphi$ be the unique subset of clauses
such that $\psi' = \psi_{\varphi'}$.
We know that $\Card{\varphi'} \leq k$.
It is easy to verify that since $\psi' \models t$,
we get that $\varphi'$ is a $k$\hy hyperpath
from $s$ to $t$.

$(\Leftarrow)$ 
Assume that there exists a $k$\hy hyperpath $\varphi' \subseteq \varphi$
from $s$ to $t$ with $\Card{\varphi'} \leq k$.
Then $\psi_{\varphi'}$ witnesses that $t$ is a $k'$-backbone
of $\psi_{\varphi}$.
Clearly, $\Card{\psi_{\varphi'}} \leq k'$.
Also, it is straightforward to verify that since $\varphi'$ is
a $k$\hy hyperpath from $s$ to $t$,
it holds that $\psi_{\varphi'} \models t$.
}
%%%
\end{proof}

\noindent Also, restricting the problem to Horn formulas
without unit clauses
unfortunately does not yield fixed-parameter tractability.

\begin{theorem}
\label{thm:bb-horn-nuhorn}
$\kbackbone[\NUHorn \cap 3\CNF]$ is \W{1}-hard.
\end{theorem}
\begin{proof}
\sloppypar We show the \W{1}-hardness of $\kbackbone[\NUHorn \cap 3\CNF]$
by reducing from \shorthyperpath{}.
Let $(\varphi,s,t,k)$ be an instance of \shorthyperpath{}.
We can assume without loss of generality that $\varphi \in 3\CNF{}$,
and that each clause in which~$t$ occurs positively
is of size $3$.
We construct an instance $(\psi_{\varphi},x_s,k)$ of \kbackbone{}.
For each $\varphi' \subseteq \varphi$ we define a formula $\psi_{\varphi'}$.
\[ \begin{array}{r l}
%  \Var{\psi_{\varphi'}} = & \SB x_v \SM v \in \Var{\varphi'} \backslash \SBs t \SEs \SE \\
  \psi_{\varphi'} = & \SB \SBs \neg x_{a}, \neg x_{b}, x_{c} \SEs \SM \SBs \neg a, \neg b, c \SEs \in \varphi', c \neq t \SE\ \cup \\
  & \SB \SBs \neg x_{a}, \neg x_{b} \SEs \SM \SBs \neg a, \neg b, t \SEs \in \varphi' \SE \\
  & \SB \SBs \neg x_{a}, x_{b} \SEs \SM \SBs \neg a, b \SEs \in \varphi' \SE \\
\end{array} \]
Clearly we have that $\psi_{\varphi} \in \Horn \cap 3\CNF$
and that $\psi_{\varphi}$ has no unit clauses.
We claim that $(\psi_{\varphi},x_s,k) \in \kbackbone{}$ if and only if
$(\varphi,s,t,k) \in \shorthyperpath{}$.
%%%
\shortversion{
A complete proof of this claim can be found in the full version of the paper.
}
%%%
\longversion{

$(\Rightarrow)$
Assume $x_s$ is a $k$\hy backbone of $\psi_{\varphi}$.
Since $\psi_{\varphi} \in \Horn$ and $\psi_{\varphi}$ has no unit clauses,
this means there exists a $\psi' \subseteq \psi_{\varphi}$ with $\Card{\psi'} \leq k$
such that $\psi' \models \neg x_s$.
Let $\varphi' \subseteq \varphi$ be the unique subset of clauses
such that $\psi' = \psi_{\varphi'}$.
We know that $\Card{\varphi'} \leq k$.
In order to show that $\varphi'$ is a hyperpath from $s$ to $t$,
we assume to the contrary that it is not.
We now define the assignment $\mu$
by letting $\mu(x_v) = \top$ for all $v \in \Var{\varphi}$
such that there exists a hyperpath $\varphi'' \subseteq \varphi'$ from $s$ to $v$
and letting $\mu(x_v) = \bot$ for all other $v \in \Var{\varphi}$.
We know that $\mu$ does not satisfy $\psi_{\varphi'}$
only if $\mu$ does not satisfy $x_a \wedge x_b$
for some $\SBs \neg x_a, \neg x_b \SEs \in \psi_{\varphi'}$
and if there exists a hyperpath from $s$
to both $a$ and $b$.
However, by the construction of $\psi_{\varphi'}$,
this can only be the case if there exists a hyperpath $\varphi'' \subseteq \varphi'$
from $s$ to $t$, which contradicts our assumption.
Thus we know that $\mu$ satisfies $\psi_{\varphi'}$
as well as $x_s$.
This is a contradiction to our previous conclusion
that $\mu$ does not satisfy $x_s$.
Therefore, we can conclude that $\varphi'$ is a hyperpath
from $s$ to $t$.
From this follows that $(\varphi,s,t,k) \in \shorthyperpath{}$.

$(\Leftarrow)$ 
Assume there exists a $k$\hy hyperpath $\varphi' \subseteq \varphi$
from $s$ to $t$.
Now consider $\psi_{\varphi'}$.
Since $\Card{\varphi'} \leq k$, we know that $\Card{\psi_{\varphi'}} \leq k$.
Also, since we know that $(\SBs a,b \SEs, t) \in \varphi'$ for some $a,b \in V$,
we know $\SBs \neg x_a, \neg x_b \SEs \in \psi_{\varphi'}$.
Now assume for an arbitrary assignment $\mu$
that $\mu \models \psi_{\varphi'}$ and $\mu \models x_s$.
By a simple inductive argument,
using the definition of $\psi_{\varphi'}$,
we then get that $\mu \models x_u$
for all $u$ for which there exists a hyperpath
from $s$ to $u$.
In particular, we get $\mu \models x_a \wedge x_b$.
However, since $\SBs \neg x_a, \neg x_b \SEs \in \psi_{\varphi'}$,
we get a contradiction to the fact that $\mu \models \psi_{\varphi'}$.
Thus we can conclude that $\psi_{\varphi'} \models \neg x_s$.
Therefore, $(\psi_{\varphi'},x_s,k) \in \kbackbone{}$.
}
%%%
\end{proof}

\paragraph{Krom formulas}
Let us now turn to the case of Krom formulas.
Restricting the problem of finding backbones in arbitrary formulas
to Krom formulas reduces the classical complexity from \coNP{}-completeness
to polynomial time solvability.
Interestingly, unlike the case for Horn formulas, the decrease in complexity
in this case also holds for finding local backbones.\footnote{In a
previous version of this paper \cite{DeHaanKanjSzeider13},
we mistakenly claimed that finding local 
backbones in Krom formulas is \W{1}-hard.}
Finding a minimum-size unsatisfiable subset of a \Krom{} formula
can be done in polynomial time \cite{BureshOppenheimMitchell06}.
This immediately implies that $\smallunsatsubset[\Krom{}]$
is polynomial-time solvable,
and therefore, by Lemma~\ref{lem:2},
$\kbackbone[\Krom{}]$ is also polynomial-time solvable
(and thus also fixed-parameter tractable).

\begin{proposition}
\label{prop:bb-krom}
$\kbackbone[\Krom{}]$ is polynomial-time solvable.
\end{proposition}

%%%
\paragraph{Hardness for finding small unsatisfiable subsets}
We would like to point out that all hardness results
for the various restrictions of \kbackbone{}
also hold for \smallunsatsubset{} under the corresponding
restrictions.
This is because the reduction in the proof of Lemma~\ref{lem:2}
works for all classes of formulas that are closed under variable instantiations.
For instance, the reduction in the proof of Lemma~\ref{lem:2}
together with Theorem~\ref{thm:bb-horn-nuhorn}
tells us that $\smallunsatsubset[\Horn{} \cap 3\CNF{}]$ is \W{1}-hard.
This does not follow from the reduction that
Fellows et al.~\cite{FellowsSzeiderWrightson06}
use to prove the \W{1}-hardness of $\smallunsatsubset{}$.
In particular, the following previously unstated
results hold.

\begin{corollary}
$\smallunsatsubset[\CCC]$ is \W{1}-hard for each
$\CCC \in \SBs \DefHorn \cap 3\CNF{}, \NUHorn{} \cap 3\CNF{} \SEs$.
%$\smallunsatsubset[\VO_d]$ is fixed-parameter tractable.
\end{corollary}

\sloppypar\noindent In fact, these fixed-parameter intractability results for $\smallunsatsubset{}$
give us the following \NP{}-hardness results.
%Interestingly, for the case of \Krom{} formulas
%this result contrasts with the known result that
%finding minimal resolution refutations for \Krom{} formulas
%can be done in polynomial time
%\cite{BureshOppenheimMitchell07,BureshOppenheimMitchell06}.

\begin{corollary}
Let $\CCC \in \SBs 3\CNF{} \cap \DefHorn{},
3\CNF{} \cap \NUHorn{} \SEs$.
Given a formula $\varphi \in \CCC$ and an integer $k$,
deciding whether $\varphi$ contains an unsatisfiable subset
of size $\leq k$ is \NP{}-hard.
\end{corollary}
\longversion{
\begin{proof}
The fpt-reductions given in the proofs of
Lemmas~\ref{lem:2}~and~\ref{lem:5} and Theorems~\ref{thm:bb-defhorn}~and~\ref{thm:bb-horn-nuhorn}
can be used as polynomial many-one reductions
from the \NP{}-hard problem of finding
a clique of certain minimum size in a graph.
\end{proof}
}

\section{Local Backbones of Formulas with Bounded Variable Occurrence}
When considering the restriction of \kbackbone{}
to formulas where variables occur a bounded number of times,
we get a fixed-parameter tractability result at last.
This fixed-parameter tractability result is closely related to
the result that \smallunsatsubset{} is fixed-parameter tractable
for instances restricted to classes of formulas that have locally bounded treewidth
\cite{FellowsSzeiderWrightson06}.
Fellows et al.~used a meta theorem to prove this.
We give a direct algorithm to solve $\smallunsatsubset[\VO_d]$
in fixed-parameter linear time.

Let $(\varphi,k)$ be an instance of $\smallunsatsubset[\VO_d]$.
The following procedure decides whether there exists
an unsatisfiable subset $\varphi' \subseteq \varphi$
of size at most $k$, and computes such a subset if it exists.
We let $\varphi^{\star} = \SB c \in \varphi \SM \Card{c} < k \SE$.
It suffices to consider subsets of $\varphi^{\star}$,
since any unsatisfiable subset $\varphi' \subseteq \varphi$ contains a
minimally unsatisfiable subset $\varphi'' \subseteq \varphi'$,
and by Tarsi's Lemma we know that $\varphi''$ contains only clauses
of size smaller than $k$.

Without loss of generality, we assume that
the incidence graph of $\varphi^{\star}$ is connected.
Otherwise, we can solve the problem by running the algorithm
on each of the connected components.
We guess a clause $c \in \varphi^{\star}$,
we let $F_1 := \SBs c \SEs$,
and we let all variables be unmarked initially.
We compute $F_{i+1}$ for $1 \leq i \leq k$ by means of the following (non-deterministic) rule:
\begin{enumerate}
  \shortversion{\vspace{-1mm}}
  \item take an unmarked variable $z \in \Var{F_i}$;
  \item guess a non-empty subset $F'_z \subseteq F_z$
  for $F_z = \SB c \in \varphi^{\star} \SM z \in \Var{c} \SE$;
  \item let $F_{i+1} := F_i \cup F'_z$;
  \item mark $z$.
  \shortversion{\vspace{-1mm}}
\end{enumerate}
If at any point all variables in $F_i$ are marked,
we stop computing $F_{i+1}$.
For any $F_i$, if $\Card{F_i} > k$ we fail.
For each $F_i$, we check whether $F_i$ is unsatisfiable.
If it is unsatisfiable, we return with $\varphi' = F_i$.
If it is satisfiable and if it contains no unmarked variables, we fail.
\shortversion{
}
\longversion{

}
It is easy to see that this algorithm is sound.
If some $\varphi' \subseteq \varphi^{\star}$ is returned,
then~$\varphi'$ is unsatisfiable and $\Card{\varphi'} \leq k$.
In order to see that the algorithm is complete,
assume that there exists some unsatisfiable $\varphi' \subseteq \varphi^{\star}$
with $\Card{\varphi'} \leq k$.
Then, since we know that the incidence graph of $F'$ is connected,
we know that $F'$ can be constructed as one of the~$F_i$ in the algorithm.

To see that this algorithm witnesses fixed-parameter linearity,
we bound its running time.
We have to execute the search process at most once for each clause
of $\varphi^{\star}$.
At each point in the execution of the algorithm,
$F_i$ contains at most $k$ variables.
Therefore, there are at most $k$ choices to take an unmarked variable $z$.
Since each variable occurs in at most $d$ clauses,
for each $F_z$ used in the rule
we know $\Card{F_z} \leq d$.
Thus, there are at most $2^{d}$ possible guesses for $F'_z$
in each execution of the rule.
Since we iterate the rule at most $k$ times,
we consider at most $(k 2^{d})^{k}$ sets $F'$,
each of size $O(k^2)$.
Thus each (un)satisfiability check can be done in $O(2^k)$ time.
Therefore, the total running time of the algorithm is $O(k^{k} 2^{d k} n)$,
for $n$ the size of the instance.

This algorithm also gives us a direct algorithm that shows that
$\kbackbone[\VO_d]$ is fixed-parameter linear.

\begin{theorem}
\label{thm:bb-vo}
$\kbackbone[\VO_d]$ is fixed-parameter linear.
\end{theorem}
\begin{proof}
The result follows directly by using the reduction in the proof of Lemma~\ref{lem:2}
in combination with the above algorithm.
\end{proof}

\section{Iterative Local Backbones}

We now consider the (parameterized) complexity
of finding iterative local backbones.
It is easy to see that \iterativebackbone{} is in \paraNP{}.
This is witnessed by a straightforward nondeterministic fpt-algorithm,
that guesses a sequence of $n$ witnesses $(\varphi_i,l_i)$
with $\Card{\varphi_i} \leq k$, and
that verifies whether $\varphi_i \subseteq \varphi|_{\{l_1,\dotsc,l_{i-1}\}}$
and whether $\varphi_i \models l_i$.

Some of the results we obtained for the problem of finding local backbones
can be carried over.
\longversion{
In all settings that yield fixed-parameter tractability for \kbackbone{}
we obtain that \iterativebackbone{} is fixed-parameter tractable as well.
}

\begin{theorem}
\label{thm:ibb-vo}
Let $\CCC$ be a class of formulas
such that $\kbackbone[\CCC]$ is fixed-parameter tractable
and $\CCC$ is closed under variable instantiation.
Then $\iterativebackbone[\CCC]$ is fixed-parameter tractable.
\end{theorem}
\begin{proof}
We give an algorithm to solve $\iterativebackbone[\CCC]$
that calls a subroutine to solve instances of $\smallunsatsubset[\CCC]$.
This algorithm is given in the form of pseudo-code
as Algorithm~\ref{alg:1}.
By the fact that $\CCC$ is closed
under variable instantiations we are able
to apply the reduction in the proof of Lemma~\ref{lem:2}.
Thus, we can assume that the question of whether some $\varphi \in \CCC$
contains an unsatisfiable subset of size at most $k$
can be solved in $f(k)\rCard{\varphi}^c$ time, for some computable function $f$
and some constant $c$.
Then, the entire algorithm runs in $O(f(k)\rCard{\varphi}^{c+2})$ time.
This proves the claim.
\end{proof}

\begin{figure}
  \shortversion{\vspace{-10pt}}
  \begin{algorithm}[H]
    \SetAlgoLined
    \SetKwInOut{Input}{input}\SetKwInOut{Output}{output}
    \SetKwData{Conseq}{conseq}
    \SetKw{Return}{return}
    \Input{an instance $(\varphi,x,k)$ of \iterativebackbone{}}
    \Output{yes iff $(\varphi,x,k) \in \iterativebackbone{}$}
    \BlankLine
    % initialize formula
    $\psi \leftarrow \varphi$%\tcp*[r]{initialize working formula}
    \longversion{\;}\shortversion{; \quad}
    % initialize list of backbone variables
    \Conseq $\leftarrow \emptyset$\; %\tcp*[r]{initialize set of backbones}
    \For{$i \leftarrow 1$ \KwTo $\Card{\Lit{\varphi}}$}{
      \ForEach{literal $l \in \Lit{\psi}$}{
        \If{$(\psi|_{\overline{l}},k) \in \smallunsatsubset{}$}{
          $\Conseq \leftarrow \Conseq \cup \{l\}$\;
        }
      }
      $\psi \leftarrow \psi|_{\Conseq}$\;
    }
    \Return{$\{x,\neg x\} \cap \Conseq \neq \emptyset$}
    \caption{Deciding \iterativebackbone{} with a \smallunsatsubset{} oracle.}
    \label{alg:1}
  \end{algorithm}
  \shortversion{\vspace{-15pt}}
\end{figure}

\longversion{
\noindent Another result that carries over from the case of finding local backbones
is the fixed-parameter intractability of finding iterative local backbones
in Horn formulas without unit clauses.
}

\begin{corollary}
\label{cor:ibb-horn}
$\iterativebackbone[\NUHorn{} \cap 3\CNF{}]$ is \W{1}-hard.
\end{corollary}
\begin{proof}
Observe that the proofs of Lemma~\ref{lem:5} and Theorem~\ref{thm:bb-horn-nuhorn}
imply that it is already \W{1}-hard to determine whether a 
formula $\varphi \in \NUHorn{} \cap 3\CNF{}$ has a subset $\varphi' \subseteq \varphi$
of size exactly $k$ witnessing that any $x \in \Var{\varphi}$ is a $k$\hy backbone.
From this, it immediately follows that determining whether
$(\varphi,x,k) \in \iterativebackbone{}$ is \W{1}-hard as well.
\end{proof}

\noindent We identify several tractable cases for
\iterativebackbone{}.
The problem of finding iterative local backbones
in definite Horn formulas is polynomial time solvable.
Interestingly, for this restriction the problem
of finding (non-iterative) local backbones remains \W{1}-hard.
Similarly, finding iterative local backbones in Krom formulas
is solvable in polynomial time as well.
This latter result already follows by Proposition~\ref{prop:bb-krom}.
We will however give an alternative (and simpler) algorithm
to find iterative local backbones in Krom formulas.
In order to show that finding iterative local backbones in
definite Horn formulas is tractable, we will use the following observation.

\begin{observation}
\label{obs:1}
Let $\varphi$ be any propositional formula,
let $l$ be any literal such that there exists a $\varphi' \subseteq \varphi$
with $\Card{\varphi'} \leq k$ and $\varphi' \models l$,
and let $\psi = \varphi|_{l}$.
Then $x \in \Var{\psi}$ is an iterative $k$\hy backbone of $\psi$
if and only if it is an iterative $k$\hy backbone of $\varphi$.
\end{observation}

\begin{theorem}
\label{thm:ibb-defhorn}
$\iterativebackbone[\DefHorn]$ is polynomial-time solvable.
\end{theorem}
\begin{proof}
We show that for any definite Horn formula $\varphi$
and any $k \geq 1$
the set of iterative $k$\hy backbones of $\varphi$
coincides with the set of variables $x \in \Var{\varphi}$
such that $\varphi \models x$.
The claim then follows, since the
entailment relation $\models$ can be decided
in linear time for definite Horn formulas \cite{DownlingGallier84}.

Fix an arbitrary integer $k \geq 1$
and an arbitrary definite Horn formula $\varphi$.
Since definite Horn formulas cannot entail negative literals,
we know that each iterative $k$\hy backbone $x$
of $\varphi$ is also a semantic consequence of $\varphi$.
Now, let $x \in \Var{\varphi}$ be an arbitrary atom
and assume that $\varphi \models x$.
So there exist variables $x_1,\dotsc,x_m \in \Var{\varphi}$
such that $x_m = x$ and for each $x_i$ we have either
(i) $\{x_i\} \in \varphi$ or (ii)~$\{ \neg x_{i_1},\dotsc,\neg x_{i_l},x_{i} \} \in \varphi$
for some $i_1 < \dotsm < i_l < i$.
We prove by induction on $m$
that each $x_i$ is an iterative $k$\hy backbone.
Take an arbitrary $x_i$.
By the induction hypothesis, we can assume that every $x_j$
for $j < i$ is an iterative $k$\hy backbone of $\varphi$.
We proceed by case distinction for the justification of $x_i$ in the sequence.
In case (i), we know that $\{ x_i \} \in \varphi$.
Therefore, it directly follows that $x_i$ is a $k$\hy backbone of $\varphi$,
and thus is an iterative $k$\hy backbone too.
In case (ii), we know that $\{ \neg x_{i_1},\dotsc,\neg x_{i_l},x_{i} \} \in \varphi$
for some $i_1 < \dotsm < i_l < i$.
By the induction hypothesis, we know that each $x_{i_j}$ is
an iterative $k$\hy backbone of $\varphi$.
By assumption, we have that $\varphi \models x_{i_j}$ for each $x_{i_j}$.
By Observation \ref{obs:1}, we get that
$x_i$ is an iterative $k$\hy backbone of $\varphi$
if and only if it is an iterative $k$\hy backbone of $\varphi_{\{x_{i_1},\dotsc,x_{i_l}\}}$.
It holds that $\{ x_i \} \in \varphi_{\{x_{i_1},\dotsc,x_{i_l}\}}$.
Thus, $x_i$ is an iterative $k$\hy backbone of $\varphi$.
\end{proof}

\begin{theorem}
\label{thm:ibb-krom}
$\iterativebackbone[\Krom]$ is polynomial-time solvable.
\end{theorem}
\begin{proof}
We show that the iterative $k$\hy backbones of a Krom formula $\varphi$
coincide with those backbones of $\varphi$ that can be identified
by iterated application of the following rule:
if the implication graph of $\varphi$ contains a path from a literal $l \in \SBs x,\neg x \SEs$
to its complement $\ol{l}$ of length at most $k$, conclude that $x$ is a backbone
and set $\varphi := \varphi|_{\ol{l}}$.
Detection of such a path can be done in polynomial time.
Also, at most $O(\Card{\Var{\varphi}})$ iterated applications of this rule
suffice to reach a fixpoint.
%All that remains to prove the result
%is to show the above correspondence claim.
All that remains is to show the correspondence.

The correspondence claim follows from the following property.
Let $l \in \Lit{\varphi}$.
If $\impl{\varphi}$ contains a path $\ol{l} \rightarrow\kleene l$
that uses at most $k$ clauses and that doubly uses $m$ of these clauses,
then there exist literals $l_1,\dotsc,l_{m+1} \in \Lit{\varphi}$
such that (i) $l_{m+1} = l$ and
(ii) for each $1 \leq i \leq m+1$ the graph $\impl{\varphi|_{\SBs l_1,\dotsc,l_{i-1} \SEs}}$
contains a path $\ol{l_i} \rightarrow\kleene l_i$ that uses at most $k$ clauses
and does not doubly use any clause.
We prove this claim by induction on $m$.
The case for $m = 0$ is trivial.
Consider the case for $m \geq 1$.
Since the path $\ol{l} \rightarrow\kleene l$ doubly uses
some clause, we know that
$\ol{l} \rightarrow\kleene a \rightarrow \ol{b} \rightarrow\kleene b \rightarrow \ol{a} \rightarrow\kleene l$,
for some $a,b \in \Lit{\varphi}$.
We can assume without loss of generality that the path $\ol{b} \rightarrow^{l} b$
does not doubly use any clause.
If this is not the case, the path $\ol{b} \rightarrow^{l} b$
contains a subpath $\ol{c} \rightarrow\kleene c$
that does not doubly uses any clauses,
and we could select $c$ instead of $b$.
Also, we know that $l \leq k$.
It is easy to see that $\impl{\varphi|_{b}}$
contains the path $\ol{l} \rightarrow\kleene a \rightarrow \ol{a} \rightarrow\kleene l$,
which uses at most $k$ clauses and doubly uses $m-1$ of these clauses.
By the induction hypothesis, we obtain that there exist $l'_1,\dotsc,l'_m$
such that $l'_m = l$ and for each $1 \leq i \leq m$
the graph $\impl{\varphi|_{\SBs l'_1,\dotsc,l'_{i-1} \SEs}}$ contains a path $\ol{l'_i} \rightarrow\kleene l'_i$
that uses at most $k$ clauses and does not doubly use any clause.
Now let $l_1 = b$ and $l_i = l'_{i-1}$ for all $2 \leq i \leq m+1$.
It is straightforward to verify that $l_1,\dotsc,l_{m+1}$ satisfy the required properties.
\end{proof}

\noindent Somewhat related to our mechanism of computing enforced assignments
via iterated $k$-backbones is the mechanism used to define
\emph{unit-refutation complete formulas of level~$k$}
\cite{GwynneKullmann13,Kullmann99}.
This mechanism is based on mappings $r_k$ from \CNF{} formulas to \CNF{} formulas.
For a nonnegative integer $k$, the mapping $r_k$ is defined inductively as follows.
In the case for $k=0$, we let $r_0(\varphi) = \SBs \bot \SEs$ if $\bot \in \varphi$,
and $r_0(\varphi) = \varphi$ otherwise.
In the case for $k>0$, we let $r_k(\varphi) = r_k(\varphi|_{l})$ if there exists
a literal $l \in \Lit{\varphi}$
such that $r_{k-1}(\varphi|_{\overline{l}}) = \SBs \bot \SEs$, and $r_k(\varphi) = \varphi$ otherwise.
In particular, the mapping $r_1$ computes the result of applying unit propagation.
Note that the result of $r_k(\varphi)$ is the application of a number of forced
assignments to $\varphi$, i.e., $r_k(\varphi) = \varphi|_L$ for some $L \subseteq \Lit{\varphi}$
such that for all $l \in L$ we have $\varphi \models l$.
We let $L^{\mtext{UC}}_k(\varphi)$ denote the set of forced literals that are computed by $r_k$,
i.e., $L^{\mtext{UC}}_k(\varphi) = L \subseteq \Lit{\varphi}$ such that $r_k(\varphi) = \varphi|_L$.
Similarly, we let $L^{\mtext{ILB}}_k(\varphi)$ denote the set of forced literals that are
found by computing iterative $k$-backbones.

The following observations relate the two mechanisms.
Let $\varphi$ be an arbitrary \CNF{} formula.
We have that $L^{\mtext{UC}}_1(\varphi) = L^{\mtext{ILB}}_1(\varphi)$.
In fact, this set contains exactly those enforced literals that can be found
by unit propagation.
Also, for any $k \geq 2$ we have that $L^{\mtext{ILB}}_k(\varphi) \subsetneq L^{\mtext{UC}}_k(\varphi)$.
The inclusion follows from the fact that each minimal subset~$\varphi'$ of size at most $k$
that enforces a literal $l$ has at most $k$ literals
(which is a direct result of Tarsi's Lemma).
Whenever $l$ is identified as an enforced literal in iterative $k$-backbone computation,
it can then also be computed by $r_k$ by first guessing $\overline{l}$,
and subsequently obtaining a contradiction for each instantiation of the other
variables in $\Var{\varphi'}$.
In order to see that the inclusion is strict,
consider the family of formulas $(\varphi_n)_{n \in \mathbb{N}}$,
where $\varphi_n = \SB \SBs \neg x_i, x_{i+1} \SEs \SM 1 \leq i < n \SE \cup
\SBs \neg x_n, \neg x_1 \SEs$.
For each $\varphi_n$, we know that $\varphi_n \models \neg x_1$.
Furthermore, we have that $\neg x_1 \in L^{\mtext{UC}}_2(\varphi_n)$,
but $x_1$ is not an iterative $k$-backbone of $\varphi_n$ for any $k < n$.

\section{Experimental Results}
In order to illustrate the relevance of the concept of local backbones and
iterative local backbones,
we provide some empirical evidence
of the distribution of (iterative) local backbones
in instances from different domains.
We considered both randomly generated instances
(3\CNF{} instances with various variable-clause ratios
around the phase transition)
and instances originating from planning~\cite{HoosStuetzle00,KautzSelman96}, circuit fault analysis~\cite{Prelotani96},
inductive inference~\cite{Prelotani96}, and bounded model checking~\cite{Strichman00}.
We considered only satisfiable instances.
%mostly taken from the SATLIB benchmark library~\cite{HoosStuetzle00}
%and the benchmark collection from the Second DIMACS Challenge~\cite{Prelotani96}.
For practical reasons,
we used a method that gives us a lower bound on the number of $k$\hy backbone variables.
By reducing the separate \kbackbone{} problems to \smallunsatsubset{},
we can use algorithms computing subset-minimal unsatisfiable subsets
to approximate the number of iterative local backbones
(we used MUSer2~\cite{BelovMarquesSilva12}).
In order to get the exact number,
we would have to compute cardinality-minimal unsatisfiable subsets,
which is difficult in practice.

%
%Note that using our approximation method allows the possibility
%to get non-monotone results
%(i.e., more iterative $k$\hy backbones found for lower orders $k$).

% Discussion of the results
The experimental results are shown in Figure \ref{fig:empirical}.
For each of the instances, we give the percentage of backbones that
are of order $k$ (dashed lines)
and the percentage of backbones
that are of iterative order $k$ (solid lines),
as well as the total number of backbones
and the total number of clauses.
There are instances with several backbones,
most of which have relatively small order.
This is the case for the instances from the domains of
planning (\emph{logistics}), circuit fault analysis (\emph{ssa7552})
and bounded model checking (\emph{bmc-ibm}).
It is worth noting that already more than 75 percent of the backbones
in all the considered \emph{bmc-ibm} instances are of iterative order 2.
We also found instances that have no backbones
of small order or of small iterative order.
This is the case for the instances from the domain
of inductive inference (\emph{ii32})
and the randomly generated instances.
Some of these instances do have backbones,
while others have no backbones at all.
\longversion{

}
It would be interesting to confirm these findings
by a more rigorous experimental investigation.

\begin{figure}[h]
\centering
\begin{subfigure}[b]{0.40\textwidth}
\begin{tikzpicture}[scale=0.85]
  \begin{axis}[
        xlabel={(iterative) order $k$},
        ylabel={\parbox{5cm}{\centering percentage of backbones that\\are of (iterative) order at most $k$}},
        ymin=0, ymax=100,
        width=7cm,
        legend style={at={(0.95,0.40)}}]
    \addplot[color=black]
      plot coordinates {
        %%% logistics.a.cnf
        (5,27.4)
        (10,72.0)
        (20,88.7)
        (30,97.2)
        (40,97.9)
        (50,98.6)
        (80,98.6)
        (100,98.6)
      };

    \addplot[color=black]
      plot coordinates {
        %%% logistics.b.cnf
        (5,29.6)
        (10,72.8)
        (20,83.9)
        (30,97.5)
        (40,100.0)
        (50,100.0)
        (80,100.0)
        (100,100.0)
      };

    \addplot[color=black]
      plot coordinates {
        %%% logistics.c.cnf
        (5,26.2)
        (10,70.1)
        (20,81.5)
        (30,96.7)
        (40,99.6)
        (50,100.0)
        (80,100.0)
        (100,100.0)
      };
    
    \addplot[color=black]
      plot coordinates {
        %%% logistics.d.cnf
        (5,0.0)
        (10,71.1)
        (20,71.1)
        (30,73.0)
        (40,73.3)
        (50,73.3)
        (80,74.2)
        (100,76.1)
      };
      
      %% NON-ITERATIVE
      
      \addplot[color=black,dashed]
      plot coordinates {
        %%% logistics.a.cnf
        (5,24.2)
        (10,46.9)
        (20,77.8)
        (30,88.5)
        (40,94.7)
        (50,97.7)
        (80,98.6)
        (100,98.6)
      };

    \addplot[color=black,dashed]
      plot coordinates {
        %%% logistics.b.cnf
        (5,22.2)
        (10,44.4)
        (20,77.7)
        (30,90.1)
        (40,97.5)
        (50,97.5)
        (80,100.0)
        (100,100.0)
      };

    \addplot[color=black,dashed]
      plot coordinates {
        %%% logistics.c.cnf
        (5,21.9)
        (10,44.5)
        (20,76.3)
        (30,88.5)
        (40,96.0)
        (50,96.0)
        (80,100.0)
        (100,100.0)
      };
    
    \addplot[color=black,dashed]
      plot coordinates {
        %%% logistics.d.cnf
        (5,15.2)
        (10,31.4)
        (20,60.6)
        (30,67.8)
        (40,72.3)
        (50,73.0)
        (80,76.8)
        (100,78.9)
      };
      
  \end{axis}
  
  \node[xshift=110pt, yshift=15pt] (a) {\scriptsize logistics};
  
\end{tikzpicture}
\end{subfigure}
\hspace{40pt}
\begin{subfigure}[b]{0.45\textwidth}
\begin{tikzpicture}[scale=0.85]
  \begin{axis}[
        xlabel={(iterative) order $k$},
        ymin=0, ymax=100,
        width=7cm,
        legend style={at={(0.95,0.40)}}]

    \addplot[color=black]
      plot coordinates {
        %%% ssa7552-038.cnf
        (5,19.2)
        (10,31.2)
        (20,48.0)
        (30,53.8)
        (40,92.3)
        (50,92.3)
        (80,100.0)
        (100,100.0)
      };

    \addplot[color=black]
      plot coordinates {
        %%% ssa7552-158.cnf
        (5,71.3)
        (10,87.9)
        (20,96.5)
        (30,100.0)
        (40,100.0)
        (50,100.0)
        (80,100.0)
        (100,100.0)
      };

    \addplot[color=black]
      plot coordinates {
        %%% ssa7552-159.cnf
        (5,69.5)
        (10,89.3)
        (20,92.5)
        (30,97.3)
        (40,97.3)
        (50,97.3)
        (80,97.3)
        (100,97.3)
      };
    
    \addplot[color=black]
      plot coordinates {
        %%% ssa7552-160.cnf
        (5,13.7)
        (10,45.6)
        (20,92.3)
        (30,100.0)
        (40,100.0)
        (50,100.0)
        (80,100.0)
        (100,100.0)
      };

    %% NON-ITERATIVE

    \addplot[color=black,dashed]
      plot coordinates {
        %%% ssa7552-038.cnf
        (5,7.6)
        (10,14.4)
        (20,44.2)
        (30,46.6)
        (40,47.1)
        (50,54.8)
        (80,69.7)
        (100,85.0)
      };

    \addplot[color=black,dashed]
      plot coordinates {
        %%% ssa7552-158.cnf
        (5,19.7)
        (10,62.7)
        (20,80.2)
        (30,85.6)
        (40,90.3)
        (50,92.6)
        (80,96.5)
        (100,96.5)
      };

    \addplot[color=black,dashed]
      plot coordinates {
        %%% ssa7552-159.cnf
        (5,22.9)
        (10,65.2)
        (20,79.1)
        (30,85.5)
        (40,91.4)
        (50,91.4)
        (80,95.1)
        (100,95.1)
      };
    
    \addplot[color=black,dashed]
      plot coordinates {
        %%% ssa7552-160.cnf
        (5,12.6)
        (10,16.2)
        (20,68.5)
        (30,85.2)
        (40,89.3)
        (50,91.3)
        (80,97.9)
        (100,97.9)
      };
      
  \end{axis}
  
  \node[xshift=110pt, yshift=15pt] (a) {\scriptsize ssa7552};
    
\end{tikzpicture}
\end{subfigure}

\begin{subfigure}[b]{0.40\textwidth}
\begin{tikzpicture}[scale=0.85]
  \begin{axis}[
        xlabel={(iterative) order $k$},
        ylabel={\parbox{5cm}{\centering percentage of backbones that\\are of (iterative) order at most $k$}},
        ymin=0, ymax=100,
        width=7cm,
        legend style={at={(0.95,0.40)}}]
    \addplot[color=black]
      plot coordinates {
        %%% bmc-ibm-2.cnf
        (5,85.5)
        (10,85.5)
        (20,87.6)
        (30,88.1)
        (40,88.9)
        (50,89.2)
        (80,100.0)
        (100,100.0)
      };

    \addplot[color=black]
      plot coordinates {
        %%% bmc-ibm-5.cnf
        (5,76.8)
        (10,78.3)
        (20,81.7)
        (30,81.8)
        (40,82.3)
        (50,85.0)
        (80,85.3)
        (100,85.7)
      };

    \addplot[color=black]
      plot coordinates {
        %%% bmc-ibm-7.cnf
        (5,90.3)
        (10,95.7)
        (20,96.6)
        (30,97.9)
        (40,98.1)
        (50,99.2)
        (80,99.8)
        (100,99.3)
      };
      
      %% NON-ITERATIVE
      
      \addplot[color=black,dashed]
      plot coordinates {
        %%% bmc-ibm-2.cnf
        (5,29.7)
        (10,42.9)
        (20,48.0)
        (30,52.6)
        (40,55.5)
        (50,60.1)
        (80,64.2)
        (100,66.0)
      };

    \addplot[color=black,dashed]
      plot coordinates {
        %%% bmc-ibm-5.cnf
        (5,17.9)
        (10,30.2)
        (20,50.2)
        (30,55.0)
        (40,57.1)
        (50,59.9)
        (80,63.6)
        (100,65.6)
      };

    \addplot[color=black,dashed]
      plot coordinates {
        %%% bmc-ibm-7.cnf
        (5,23.1)
        (10,41.3)
        (20,55.6)
        (30,59.4)
        (40,63.8)
        (50,68.5)
        (80,70.6)
        (100,73.3)
      };
      
  \end{axis}
  
    \node[xshift=110pt, yshift=15pt] (a) {\scriptsize bmc-ibm};

\end{tikzpicture}
\end{subfigure}
\hspace{40pt}
\begin{subfigure}[b]{0.45\textwidth}
\begin{tikzpicture}[scale=0.85]
  \begin{axis}[
        xlabel={(iterative) order $k$},
        ymin=0, ymax=100,
        width=7cm,
        legend style={at={(0.95,0.40)}}]

    \addplot[color=black,mark=x]
      plot coordinates {
        %%% ii32*
        (5,0)
        (10,0)
        (20,0)
        (30,0)
        (40,0)
        (50,0)
        (80,0)
        (100,0)
      };
    
    \addplot[color=black,mark=x]
      plot coordinates {
        %%% random-*
        (5,0)
        (10,0)
        (20,0)
        (30,0)
        (40,0)
        (50,0)
        (80,0)
        (100,0)
      };

  \end{axis}
  
\node[xshift=110pt, yshift=20pt] (a) {\parbox{18pt}{\scriptsize ii32\\random}};
  
\end{tikzpicture}
\end{subfigure}
\caption{Percentage of backbones that are of order at most $k$ (dashed)
and of iterative order at most $k$ (solid),
for \SAT{} instances from
planning (\emph{logistics.[a--d]}, 828--4713 variables, 6718--21991 clauses, 437--838 backbones),
circuit fault analysis (\emph{ssa7552-[038,158--160]}, 1363--1501 variables, 3032--3575 clauses, 405--838 backbones),
bounded model checking (\emph{bmc-ibm-[2,5,7]}, 2810--9396 variables, 11683--41207 clauses, 405--557 backbones),
inductive inference (\emph{ii32[b--e][1--3]}, 222--824 variables, 1186--20862 clauses, 0--208 backbones)
and random 3\SAT{} instances (\emph{random}, 200 variables, 820--900 clauses, 1--131 backbones).
}
\label{fig:empirical}
\end{figure}
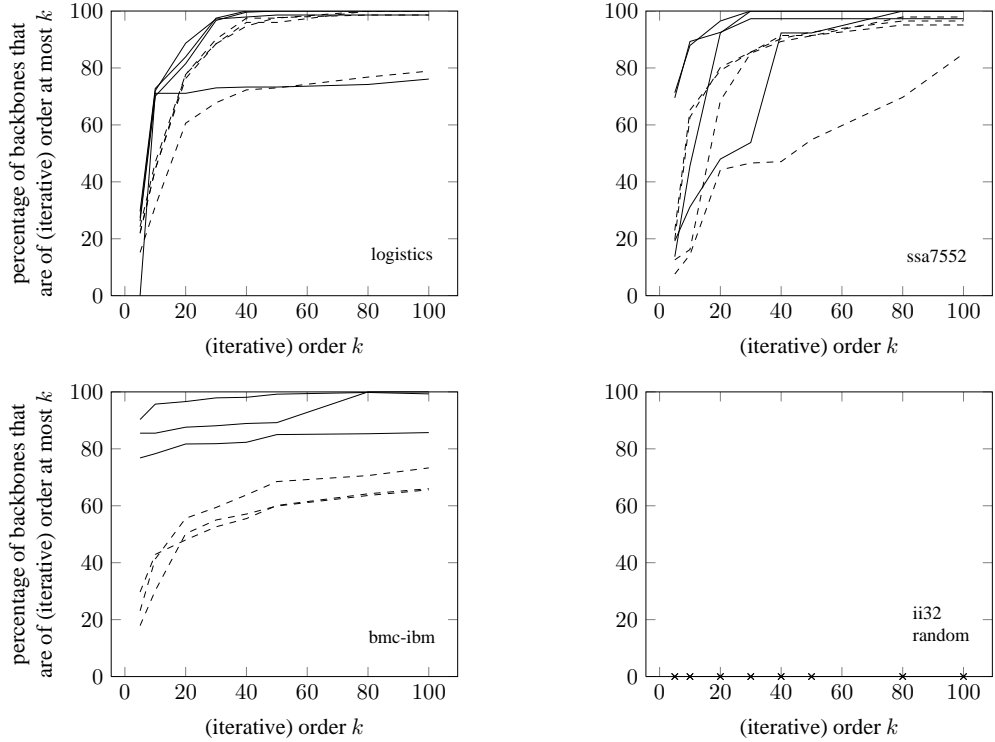

\section{Conclusions}
We have drawn a detailed complexity map of
the problem of finding local backbones
and iterative local backbones,
in general and for formulas
from restricted classes.
Additionally, we have provided some first empirical results
on the distribution of (iterative) local backbones
in some benchmark \SAT{} instances.
We found that in structured instances from different domains
backbones are of quite low (iterative) order.
This suggests that the notions of local backbones
and iterative local backbones can be used to identify
structure in \SAT{} instances.

Some of our findings are somewhat surprising.
(1) Finding local backbones in Horn formulas
is fixed-parameter intractable,
whereas backbones for this class of formulas
can be found in polynomial time.
(2) In certain cases finding iterative local backbones
is computationally easier than finding (non-iterative) local backbones.
(3) Local backbones and iterative local backbones seem to
be a better indicator of structure than backbones.
Random instances do have backbones, but these
are of high order and iterative order.

Backbones and local backbones are implied unit clauses.
It might be interesting to extend our investigation to implied
clauses of larger fixed size, binary clauses in particular.

\vfill
\pagebreak\clearpage
\DeclareRobustCommand{\DE}[3]{#3}
\bibliographystyle{abbrv}
%\bibliography{literature}

\end{document}